\documentclass[a4paper,10pt,twoside,twocolumn,openany,final]{memoir}
\usepackage{amsmath,amssymb,amsthm}
\usepackage{graphicx}

\settrimmedsize{\stockheight}{\stockwidth}{*}

\setlrmarginsandblock{17.0mm}{18.0mm}{*}
\setulmarginsandblock{25.0mm}{23.5mm}{*}
\setheadfoot{9mm}{1em}

\checkandfixthelayout

\setsecnumdepth{subsection}

\setsecnumformat{{\csname pre#1\endcsname \csname the#1\endcsname}\csname post#1\endcsname}
\setsecheadstyle{\scshape\centering\large}
\setsubsecheadstyle{\itshape}
\setbeforesecskip{3.0mm}
\setaftersecskip{1.0mm}
\setaftersubsecskip{0.5mm}
\setbeforesubsecskip{0.8mm}

\usepackage{mathptmx}

\DeclareMathAlphabet{\mathcal}{OMS}{cmsy}{m}{n}
\SetSymbolFont{operators}{bold}{OT1}{txr}{bx}{n}
\SetSymbolFont{letters}{bold}{OML}{txmi}{bx}{it}
\SetSymbolFont{symbols}{bold}{OMS}{txsy}{bx}{n}
\SetSymbolFont{largesymbols}{bold}{OMX}{txex}{bx}{n}

\usepackage{url}

\newcommand{\Var}{\mathrm{Var}}
\newcommand{\Cov}{\mathrm{Cov}}
\newcommand{\Cor}{\mathrm{Cor}}

\usepackage{bm}

\allowdisplaybreaks[3]

\urldef\rmmail\url{rmori@sys.i.kyoto-u.ac.jp}
\urldef\ttmail\url{tt@i.kyoto-u.ac.jp}

\title{New Generalizations of the Bethe Approximation via Asymptotic Expansion}

\thanksmarkseries{fnsymbol}

\author{
Ryuhei~Mori~~and~~Toshiyuki Tanaka
\thanks{
Graduate School of Informatics,
Kyoto University,
Kyoto, 606--8501, Japan
(email: \rmmail,\; \ttmail).
The work of R. Mori was supported by the Grant-in-Aid for Scientific Research for JSPS Fellows (22$\cdot$5936).}%
}

\date{}

\theoremstyle{plain}
\newtheorem{theorem}{Theorem}
\newtheorem{lemma}[theorem]{Lemma}

\theoremstyle{definition}
\newtheorem{definition}[theorem]{Definition}
\newtheorem{example}[theorem]{Example}
\theoremstyle{remark}
\newtheorem{remark}[theorem]{Remark}

\makepagestyle{mysimple}
\makeevenhead{mysimple}{}{}{\thepage}
\makeoddhead{mysimple}{}{}{\thepage}
\makeevenfoot{mysimple}{}{}{}
\makeoddfoot{mysimple}{}{}{}

\makeoddhead{myheadings}{}{}{%
\vbox{%
\scriptsize%
\hbox{The 35th Symposium on Information Theory}
\hbox{and its Applications (SITA2012)}
\hbox{Beppu, Oita, Japan, Dec.\ 11--14, 2012}%
}%
}

\begin{document}
\maketitle
\thispagestyle{myheadings}

\begin{abstract}
The Bethe approximation, discovered in statistical physics, gives an efficient algorithm
called belief propagation (BP) for approximating a partition function.
BP empirically gives an accurate approximation for many problems, e.g., low-density parity-check codes, compressed sensing, etc.
Recently, Vontobel gives a novel characterization of the Bethe approximation using graph cover.
In this paper, a new approximation based on the Bethe approximation is proposed.
The new approximation is derived from Vontobel's characterization using graph cover,
and expressed by using the edge zeta function,
which is related with the Hessian of the Bethe free energy as shown by Watanabe and Fukumizu.
On some conditions, it is proved that the new approximation is asymptotically better than the Bethe approximation.
\end{abstract}
\begin{center}
\vspace{-3mm}
Keywords
\end{center}
\vspace{-4mm}
{\small
Bethe approximation, belief propagation, partition function, graphical model, edge zeta function, asymptotic expansion.
}

\section{Introduction}
Calculation of partition functions is one of the central problems in statistical physics.
The Bethe approximation is an empirically successful efficient approximation
for many problems~\cite{mezard2009information}.
A few results about asymptotic exactness of the Bethe approximation are recently known.
Recently, Vontobel shows a novel characterization of the Bethe approximation using graph cover~\cite{vontobel2010counting}.
In this paper, a series of generalizations of the Bethe approximation is shown from Vontobel's characterization of the Bethe free energy
via asymptotic expansion.

\section{Factor graph and the Bethe free energy}
Let $\mathcal{X}$ be a finite set.
A factor graph is a bipartite graph consisting of variable nodes and factor nodes.
A factor graph including $N$ variable nodes defines a probability measure on $\mathcal{X}^N$.
The set of variable nodes, the set of factor nodes and the set of edges are denoted by $V$, $F$ and $E\subseteq V\times F$, respectively.
The neighborhoods of a factor node $a\in F$ and a variable node $i\in V$ are denoted by $\partial a \subseteq V$ and $\partial i \subseteq F$, respectively.
The degrees of variable node $i$ and factor node $a$ are denoted by $d_i$ and $d_a$, respectively.
For each factor node $a$, there exists a corresponding non-negative function $f_a: \mathcal{X}^{d_a}\to \mathbb{R}_{\ge 0}$.
The probability mass function defined by the factor graph is
\begin{align*}
p(\bm{x}) &= \frac1Z \prod_{a\in F} f_a(\bm{x}_{\partial a}),&
Z &:= \sum_{\bm{x}\in\mathcal{X}^N} \prod_{a\in F} f_a(\bm{x}_{\partial a}).
\end{align*}
The constant $Z$ for the normalization is called the {partition function}.
The calculation of $Z$ is generally a \#P-complete problem.
Approximating $Z$ is one of the central problems of statistical physics and information theory~\cite{mezard2009information}.

In the following, the Bethe free energy is defined using the variational method~\cite{yedidia2005constructing}.
Let $p(\bm{x})$ be a probability mass function defined by a factor graph.
From the following equation for any probability measure $q(\bm{x})$ on $\mathcal{X}^N$,
\begin{align*}
&D(q\| p) := \sum_{\bm{x}\in\mathcal{X}^N} q(\bm{x})\log \frac{q(\bm{x})}{p(\bm{x})}\\
&\quad = \log Z - \sum_{a\in F} \sum_{\bm{x}\in\mathcal{X}^N} q(\bm{x})\log f_a(\bm{x}_{\partial a})
+ \sum_{\bm{x}\in\mathcal{X}^N} q(\bm{x})\log q(\bm{x})\\
&\quad =: \log Z + \mathcal{U}_\mathrm{Gibbs}(q) - \mathcal{H}_\mathrm{Gibbs}(q)
=: \log Z + F_\mathrm{Gibbs}(q)
\end{align*}
one obtains the equation $-\log Z = F_\mathrm{Gibbs}(p) = \min_{q\in\mathcal{P}(\mathcal{X}^N)} F_\mathrm{Gibbs}(q)$
where $\mathcal{P}(\mathcal{X}^N)$ denotes the set of probability measures on $\mathcal{X}^N$.
The Bethe free energy, defined in the following, is an approximation of the Gibbs free energy, which is $F_\mathrm{Gibbs}(q)$
in the above equation.
\begin{definition}[Bethe free energy]
For $\{b_i\in\mathcal{P}(\mathcal{X})\}_{i\in V}$ and $\{b_a\in\mathcal{P}(\mathcal{X}^{d_a})\}_{a\in F}$ satisfying the reducibility 
\begin{equation}\label{eq:marginal}
\sum_{\bm{x}_{\partial a} \in\mathcal{X}^{d_a}, x_i = z} b_a(\bm{x}_{\partial a}) = b_i(z),\quad \forall (i,a)\in E
\end{equation}
the Bethe free energy is defined as
\begin{align*}
F_\mathrm{Bethe}(\{b_i\}, \{b_a\}) &:= \mathcal{U}_\mathrm{Bethe}(\{b_i\}, \{b_a\}) - \mathcal{H}_\mathrm{Bethe}(\{b_i\}, \{b_a\})
\end{align*}
where
\begin{align*}
\mathcal{U}_\mathrm{Bethe}(\{b_i\}, \{b_a\})
&:= -\sum_{a\in F} \sum_{\bm{x}_{\partial a}\in\mathcal{X}^{d_a}} b_a(\bm{x}_{\partial a})\log f_a(\bm{x}_{\partial a})\\
\mathcal{H}_\mathrm{Bethe}(\{b_i\}, \{b_a\})
&:= -\sum_{a\in F} \sum_{\bm{x}_{\partial a}\in\mathcal{X}^{d_a}} b_a(\bm{x}_{\partial a})\log b_a(\bm{x}_{\partial a})\\
&\quad +\sum_{i\in V} (d_i-1)\sum_{x_{i}\in\mathcal{X}} b_i(x_i)\log b_a(x_i).\qed
\end{align*}
\end{definition}
Since the Bethe free energy is considered as an approximation of the Gibbs free energy,
the minimum of the Bethe free energy is regarded as an approximation of $-\log Z$.
The stationary condition of the Bethe free energy is
\begin{equation}
\begin{split}
b_a(\bm{x}) &= \frac1{Z_a(\{m_{i\to a}\}_{i\in\partial a})} f_a(\bm{x})\prod_{i\in\partial a} m_{i\to a}(x_i)\\
b_i(x) &= \frac1{Z_i(\{m_{a\to i}\}_{a\in\partial i})} \prod_{a\in\partial i} m_{a\to i}(x)\\
&= \frac1{Z_{i,a}(m_{a\to i}, m_{i\to a})} m_{a\to i}(x)m_{i\to a}(x),\quad \forall a\in\partial i.
\end{split}
\label{eq:biba}
\end{equation}
where 
$\{(m_{a\to i}, m_{i\to a})\in\mathcal{P}(\mathcal{X})^2\}_{(i,a)\in E}$ satisfies
the belief propagation (BP) equations
\begin{equation}
\begin{split}
m_{i\to a}(x) &\propto \prod_{b\in\partial i\setminus a} m_{b\to i}(x)\\
m_{a\to i}(x) &\propto \sum_{\bm{x}\in\mathcal{X}^{d_a}, x_i=x}f_a(\bm{x})\prod_{j\in\partial a\setminus i} m_{j\to a}(x_j)
\end{split}
\label{eq:bp}
\end{equation}
and where
$Z_a(\{m_{i\to a}\}_{i\in\partial a})$, $Z_i(\{m_{a\to i}\}_{a\in\partial i})$ and $Z_{i,a}(m_{a\to i}, m_{i\to a})$
are the constants for the normalizations
\begin{align*}
Z_a(\{m_{i\to a}\}_{i\in\partial a}) &:= \sum_{\bm{x}\in\mathcal{X}^{d_a}}f_a(\bm{x})\prod_{i\in\partial a} m_{i\to a}(x_i)\\
Z_i(\{m_{a\to i}\}_{a\in\partial i}) &:= \sum_{x\in\mathcal{X}} \prod_{a\in\partial i} m_{a\to i}(x)\\
Z_{i,a}(m_{a\to i}, m_{i\to a}) &:= \sum_{x\in\mathcal{X}} m_{a\to i}(x)m_{i\to a}(x).
\end{align*}
By substituting~\eqref{eq:biba} to the definition of the Bethe free energy and using~\eqref{eq:bp},
one obtains the representation
\begin{align}
&F_\mathrm{Bethe}(\{m_{i\to a}\}, \{m_{a\to i}\})
=
-\sum_{i\in V} \log Z_i(\{m_{a\to i}\}_{a\in\partial i})
\label{eq:Bethe}\\
&\quad -\sum_{a\in F} \log Z_a(\{m_{i\to a}\}_{i\in\partial a})
 +\sum_{(i,a)\in E} \log Z_{i,a}(m_{i\to a}, m_{a\to i}).\nonumber
\end{align}
Note that the stationary condition of~\eqref{eq:Bethe} is also given by the BP equations~\eqref{eq:bp}.
Since the Bethe free energy is an approximation of $-\log Z$, the Bethe partition function is defined as
$Z_\mathrm{Bethe}(\{m_{i\to a}\}, \{m_{a\to i}\}) := \exp\{-F_\mathrm{Bethe}(\{m_{i\to a}\}, \{m_{a\to i}\})\}$
i.e.,
\begin{align*}
&Z_\mathrm{Bethe}(\{m_{i\to a}\}, \{m_{a\to i}\}) = 
\prod_{i\in V} Z_i(\{m_{a\to i}\}_{a\in\partial i})\\
&\quad\cdot\prod_{a\in F} Z_a(\{m_{i\to a}\}_{i\in\partial a})
\prod_{(i,a)\in E} \frac1{Z_{i,a}(m_{i\to a}, m_{a\to i})}
\end{align*}
for $\{(m_{a\to i}, m_{i\to a})\}_{(i,a)\in E}$ satisfying~\eqref{eq:bp}.
We also use the notation $Z_\mathrm{Bethe}:= \min_{\{m_{i\to a}\}, \{m_{a\to i}\}} Z_\mathrm{Bethe}(\{m_{i\to a}\}, \{m_{a\to i}\})$.
We assume in this paper that all variable nodes have the same alphabet $\mathcal{X}$ 
only for the simplicity.
All results in this paper can be generalized to the case that each variable node $i\in V$ has distinct alphabet $\mathcal{X}_i$.

\section{Exponential family}
In this section, exponential family of probability distributions is introduced.
The exponential family is a class of parametric family of distributions.
A family of probability distributions having a parameter $\bm{\theta}\in\Theta\subseteq\mathbb{R}^d$ for some $d\in\mathbb{N}$
is called a parametric family of distributions.
In this paper, the existence of the first and the second derivatives of probability density (mass) function $p(x\mid \bm{\theta})$
with respect to the parameter is assumed.
\begin{definition}[Fisher information matrix]
For a parametric family of distributions,
the Fisher information matrix is defined for each parameter $\theta\in\Theta$ as
\begin{equation*}
\mathcal{J}(\bm{\theta})_{k,l} := 
\mathbb{E}\bigg[\frac{\partial \log p(X\mid \bm{\theta})}{\partial \theta_k}\frac{\partial \log p(X\mid \bm{\theta})}{\partial \theta_l}\bigg]
\end{equation*}
where $X\sim p(X\mid \bm{\theta})$ and where ${B}_{k,l}$ denotes the $(k,\;l)$-element of a matrix $B$.
\end{definition}

\begin{definition}[Exponential family]
The exponential family is a parametric family of probability distributions
whose probability density (mass) functions can be expressed in the form
\begin{equation*}
p(x\mid \bm{\theta})=
\exp\bigg\{\sum_{k=1}^d t_k(x)\theta_k - \psi(\bm{\theta})\bigg\}
\end{equation*}
using a set of functions $\{t_k:\mathcal{X}\to\mathbb{R}\}_{k=1,2,\dotsc,d}$ called a sufficient statistic.
Here,
\begin{equation*}
\psi(\bm{\theta}) := \log \sum_{x\in\mathcal{X}}\exp\bigg\{\sum_{k=1}^d t_k(x)\theta_k \bigg\}.
\end{equation*}
\end{definition}
The parameter $\bm{\theta}$ of an exponential family is called a natural parameter.
The parameterization with $\{\eta_k:=\mathbb{E}[t_k(X)]\}_{k=1,\dotsc,d}$ is 
also useful, and the parameter is called an expectation parameter.

\begin{example}[Multinomial distribution]\label{exm:multi}
The family of distributions on a finite set $\mathcal{X}=\{0,1,\dotsc,q-1\}$ can be regarded as
the exponential family with a sufficient statistic $\{t_{x}(x')=\delta_{x,x'}\}_{x\in\mathcal{X}\setminus 0}$.
In this case, $\eta_x=p(x\mid \bm{\theta})$ for $x\in\mathcal{X}\setminus 0$.
\end{example}

Let $\nabla^2 Q(\bm{\theta})$ be the Hessian matrix of $Q$ at $\bm{\theta}$.
Let $\Var[\bm{t}(X)]$ be the covariance matrix of $\{t_k(X)\}_{k=1,\dotsc,d}$ 
with respect to $p(x\mid\bm{\theta})$.
Then, it is easy to see that $\Var[\bm{t}(X)]=\mathcal{J}(\bm{\theta})$.

\begin{lemma}\label{lem:fisher}
It holds
$\mathcal{J}(\bm{\theta}) = 
\mathcal{J}(\bm{\eta})^{-1}$
where $\mathcal{J}(\bm{\theta})$ and $\mathcal{J}(\bm{\eta})$ 
are the Fisher information matrices with respect to the natural parameter $\bm{\theta}$
and the expectation parameter $\bm{\eta}$.
Furthermore, it holds
$-\nabla^2 H(\bm{\eta}) = \mathcal{J}(\bm{\eta})$ and
$\nabla^2 \psi(\bm{\theta}) = \mathcal{J}(\bm{\theta})$
where $H(\bm{\eta}):=-\sum_{x\in\mathcal{X}} p(x\mid\bm{\eta})\log p(x\mid\bm{\eta})$
is the Shannon entropy.
\end{lemma}

\section{Characterization of the Bethe free energy using the loop calculus}
\subsection{Linear transforms of the factor functions}
In this section, a characterization of the Bethe free energy is described which is obtained by Chernyak and Chertkov~\cite{chernyak2007loop}.
Assume that a function $f_a(\bm{x}):\mathcal{X}^{d_a}\to\mathbb{R}_{\ge 0}$ is represented by using
$\hat{f}_a: \mathcal{Y}^{d_a}\to\mathbb{R}_{\ge 0}$ and $\phi_{i,a}:\mathcal{X}\times\mathcal{Y}\to\mathbb{R}_{\ge 0}$ as 
\begin{equation}
f_a(\bm{x}_{\partial a})
= \sum_{\bm{y}_{\partial a}\in\mathcal{Y}^{d_a}}
\hat{f}_a(\bm{y}_{\partial a}) \prod_{i\in\partial a} \phi_{i,a}(x_i, y_i)
\label{eq:yrd}
\end{equation}
for some set $\mathcal{Y}$.
Let $\delta(x,z)$ be a function which takes 1 if $x=z$ and takes 0 otherwise.
When $\mathcal{Y}=\mathcal{X}$, this representation is obtained by letting
\begin{equation*}
\hat{f}_a(\bm{y}_{\partial a})
= \sum_{\bm{x}_{\partial a}\in\mathcal{X}^{d_a}}
f_a(\bm{x}_{\partial a}) \prod_{i\in\partial a} \hat{\phi}_{i,a}(y_i, x_i)
\end{equation*}
for $\{(\phi_{i,a},\hat{\phi}_{i,a})\}_{i\in\partial a}$ satisfying
\begin{align*}
\sum_{y\in\mathcal{X}} \phi_{i,a}(x,y)\hat{\phi}_{i,a}(y,z) &= \delta(x,z)
\end{align*}
or equivalently
\begin{align}
\sum_{x\in\mathcal{X}} \hat{\phi}_{i,a}(y,x)\phi_{i,a}(x,w) &= \delta(y,w).
\label{eq:inverse1}
\end{align}

The partition function can be rewritten by using the transform~\eqref{eq:yrd} as
\begin{align}
Z 
&= \sum_{\bm{y}\in\mathcal{Y}^{|E|}} \prod_{a\in F} \hat{f}_a(\bm{y}_{\partial a})
\prod_{i\in V} \left(\sum_{x\in\mathcal{X}} \prod_{a\in\partial i} \phi_{i,a}(x,y_{i,a})\right).
\label{eq:expansion}
\end{align}
In this representation, the variables $\bm{y}\in\mathcal{Y}^E$ are associated with edges of the original factor graph.
Both the variable nodes and factor nodes in the original graph can be regarded as factor nodes in the new representation~\eqref{eq:expansion}.
The transform of representation gives a unified way to understand many equations, e.g., the MacWilliams identity~\cite{forney2011partition}.

\subsection{Bethe transforms and characterization of Bethe free energy}
Let $\mathcal{Y}=\mathcal{X}$.
We now consider the following additional conditions on $\{(\phi_{i,a}, \hat{\phi}_{i,a})\}_{(i,a)\in E}$.
For each $i\in V$ and $a\in F$,
\begin{equation}
\begin{split}
\sum_{x\in\mathcal{X}} \prod_{a\in\partial i} \phi_{i,a}(x,y_a) &= 0,\quad \exists! b\in\partial i, y_b\ne 0\\
\sum_{\bm{x}_{\partial a}\in\mathcal{X}^{d_a}}
f_a(\bm{x}_{\partial a}) \prod_{i\in\partial a} \hat{\phi}_{i,a}(y_i, x_i)
&= 0,\quad \exists! j\in\partial a, y_j\ne 0.
\end{split}
\label{eq:condBethe0}
\end{equation}
On these conditions, if the subset $\{(i,a)\in E\mid y_{i,a} \ne 0\}\subseteq E$ of edges generates degree-one variable nodes or degree-one factor nodes,
the term in~\eqref{eq:expansion} corresponding to $\bm{y}$ is zero.
Hence, in~\eqref{eq:expansion}, we only have to take the sum over $\bm{y}\in\mathcal{X}^{|E|}$
satisfying $\{(i,a)\in E\mid y_{i,a}\ne 0\}\in G$ where
$G$ is the set of generalized loops defined as
$G := \left\{E'\subseteq E\mid d_o(E') \ne 1, \forall o\in V\cup F\right\}$.
Here, $d_i(E'):=|\{(i,a)\in E'\mid a\in\partial i\}|$ and
$d_a(E'):=|\{(i,a)\in E'\mid i\in\partial a\}|$.
The conditions~\eqref{eq:inverse1} and \eqref{eq:condBethe0}
are equivalent to the condition~\eqref{eq:inverse1} together with
\begin{equation}
\begin{split}
\hat{\phi}_{i,a}(0,x) &= \frac1{\sum_{x\in\mathcal{X}}\prod_{b\in\partial i} \phi_{i,b}(x,0)} \prod_{b\in\partial i\setminus a} \phi_{i,b}(x,0)
\\
\phi_{i,a}(x,0) &= \frac1{\hat{f}_a(\text{all } 0)}
 \sum_{\substack{\bm{x}_{\partial a}\in\mathcal{X}^{d_a},\\ x_i=x}} f_a(\bm{x}_{\partial a})\prod_{j\in\partial a\setminus i}\hat{\phi}_{j,a}(0,x_j).
\end{split}
\label{eq:condBethe1}
\end{equation}

For $\{(m_{a\to i}, m_{i\to a})\}_{(i,a)\in E}$ which satisfies the BP equations~\eqref{eq:bp}, $\bigl(\phi_{i,a}(x, 0), \hat{\phi}_{i,a}(0, x)\bigr) 
= \bigl(c_{i,a}m_{a\to i}(x), \hat{c}_{i,a}m_{i\to a}(x)\bigr)$ provides a solution of~\eqref{eq:condBethe1}
where $c_{i,a}\hat{c}_{i,a}=1/Z_{i,a}(m_{i\to a}, m_{a\to i})$.
In this case, the contribution of all-zero assignment in~\eqref{eq:expansion} is
the Bethe partition function $Z_{\mathrm{Bethe}}(\{m_{i\to a}\}, \{m_{a\to i}\})$.

For the binary case i.e., $\mathcal{X}=\{0,1\}$,
$\bigl(\phi_{i,a}(x, 1), \hat{\phi}_{i,a}(1, x)\bigr) 
= \bigl((-1)^{\bar{x}}c_{i,a}m_{i\to a}(\bar{x}), (-1)^{\bar{x}}\hat{c}_{i,a}m_{a\to i}(\bar{x})\bigr)$
satisfies~\eqref{eq:inverse1}.
In this case, one obtains the following lemma by substituting
the above values of $\{\phi_{i,a}, \hat{\phi}_{i,a}\}_{(i,a)\in E}$ to~\eqref{eq:expansion}.
Let $\langle\cdot\rangle_{p}$ be an expectation with respect to a probability mass function $p$.
\begin{lemma}[\cite{chertkov2006loop}, \cite{sudderth2008loop}]\label{lem:loop}
Assume that the alphabet is binary, i.e., $\mathcal{X}=\{0,1\}$.
Let $m_i:=\langle X_i\rangle_{b_i}= b_i(1)$.
For any stationary point $(\{b_i\}, \{b_a\})$ of the Bethe free energy,
\begin{equation}
\frac{Z}{Z_\mathrm{Bethe}(\{b_i\},\{b_a\})} = \sum_{E'\subseteq E} \mathcal{Z}(E')
\label{eq:loop}
\end{equation}
where
\begin{align*}
\mathcal{Z}(E')&:=
\prod_{i\in V}\left\langle\left(\frac{X_i-m_i}{\sqrt{\langle(X_i-m_i)^2\rangle_{b_i}}}\right)^{d_i(E')}\right\rangle_{b_i}\\
&\qquad\cdot\prod_{a\in F}\left\langle\prod_{i\in \partial a,\, (i,a)\in E'}\frac{X_i-m_i}{\sqrt{\langle(X_i-m_i)^2\rangle_{b_i}}}\right\rangle_{b_a}.
\end{align*}
\end{lemma}

For non-binary cases, the conditions~\eqref{eq:inverse1} and \eqref{eq:condBethe0}
do not fix $\{(\phi_{i,a}, \hat{\phi}_{i,a})\}_{(i,a)\in E}$ uniquely.
In~\cite{chernyak2007loop}, it is suggested to use loop calculus iteratively for each $\mathcal{Z}(E')$.
In this paper, on the other hand, we propose an explicit choice of $\{\phi_{i,a}, \hat{\phi}_{i,a}\}$.
As shown in Example~\ref{exm:multi}, the family of multinomial distributions can be regarded as an exponential family.
Let $\bm{\theta}_i$ and $\bm{\eta}_i$ be a natural parameter and an expectation parameter of $b_i$, respectively.
Then, $\phi_{i,a}(x, y)$ and $\hat{\phi}_{i,a}(x, y)$ for $x\in\mathcal{X}$ and $y\in \mathcal{X}\setminus 0$ are chosen as
\begin{align*}
\frac{\phi_{i,a}(x, y)}{c_{i,a}m_{a\to i}(x)}
&= \frac{\partial \log b_i(x)}{\partial \eta_{i,y}},&
\frac{\hat{\phi}_{i,a}(x, y)}{\hat{c}_{i,a}m_{i\to a}(x)}
&= \frac{\partial \log b_i(x)}{\partial \theta_{i,y}}.
\end{align*}
The partial derivatives in the first and second equations in the above 
are those with respect to
the coordinate systems $\{\eta_{i,y}\}_{y\in\mathcal{X}\setminus 0}$ and $\{\theta_{i,y}\}_{y\in\mathcal{X}\setminus 0}$, respectively.
One can easily confirm that these $\{\phi_{i,a}, \hat{\phi}_{i,a}\}$ satisfy the condition~\eqref{eq:inverse1} as follows.
For $w\in \mathcal{X}\setminus 0$, it holds
\begin{align*}
\sum_{x\in\mathcal{X}} \hat{\phi}_{i,a}(0,x)\phi_{i,a}(x,w) &=
\sum_{x\in\mathcal{X}} b_i(x)
\frac{\partial \log b_i(x)}{\partial \theta_{i,w}}=0.
\end{align*}
Similarly, $\sum_{x\in\mathcal{X}} \hat{\phi}_{i,a}(y,x)\phi_{i,a}(x,0) = 0$ for any $y\in\mathcal{X}\setminus 0$.
For $y, w\in \mathcal{X}\setminus 0$, it holds
\begin{align*}
\sum_{x\in\mathcal{X}} \hat{\phi}_{i,a}(y,x)\phi_{i,a}(x,w) &=
\sum_{x\in\mathcal{X}} 
b_i(x)
\frac{\partial \log b_i(x)}{\partial \eta_{i,y}}
\frac{\partial \log b_i(x)}{\partial \theta_{i,w}}\\
&=
\sum_{x\in\mathcal{X}} 
\frac{\partial b_i(x)}{\partial \eta_{i,y}}
\left[t_{i, w}(x) - \eta_{i, w}\right]\\
&=
\frac{\partial \eta_{i,w}}{\partial \eta_{i,y}}
-
\eta_{i,w}
\sum_{x\in\mathcal{X}} 
\frac{\partial b_i(x)}{\partial \eta_{i,y}}
 = \delta(y, w).
\end{align*}
Then, one obtains the following theorem.
\begin{theorem}[Loop calculus for non-binary alphabet]\label{thm:non-binary}
For any stationary point $(\{b_i\}, \{b_a\})$ of the Bethe free energy,~\eqref{eq:loop} holds
where
\begin{align*}
\mathcal{Z}(E')&:=\sum_{\bm{y}\in(\mathcal{X}\setminus\{0\})^{|E'|}}
\prod_{i\in V}
\left\langle \prod_{a\in\partial i, (i,a)\in E'} \frac{\partial \log b_i(X_i)}{\partial \eta_{i,y_{i,a}}}\right\rangle_{b_i}\\
&\quad\cdot\prod_{a\in F}
\left\langle \prod_{i\in\partial a, (i,a)\in E'} \frac{\partial \log b_i(X_i)}{\partial \theta_{i,y_{i,a}}}\right\rangle_{b_a}.
\end{align*}
\end{theorem}
This theorem is useful for understanding the approximations to be proposed in Sec.~\ref{sec:new}.

\section{Edge zeta function and the determinant of Hessian of the Bethe free energy}
\subsection{Edge zeta function}
In this section, the edge zeta function and Watanabe-Fukumizu formula are introduced~\cite{watanabe2010thesis}.
\begin{definition}
For $(i,a)\in E$ and $(j, b)\in E$,
$(i\to a)\rightharpoonup (j\to b) \overset{\mathrm{def}}{\iff} j\in\partial a, i\ne j, a\ne b$.
\end{definition}

\begin{definition}[Prime cycle]
The closed walk $e_1\rightharpoonup e_2\rightharpoonup \dotsc \rightharpoonup e_n \rightharpoonup e_1$ is said to be a prime cycle
if and only if it cannot be expressed as a power of another walk.
Prime cycles are identified up to cyclic permutations.
\end{definition}

\begin{definition}[Edge zeta function~\cite{watanabe2010thesis}]
Let $r_{i\to a}$ be a natural number associated with an edge $(i, a)\in E$ and
$u_{(i\to a), (j\to b)}$ be an $r_{i\to a}\times r_{j\to b}$ matrix for $(i\to a)\rightharpoonup (j\to b)$.
Then, the edge zeta function is defined as
\begin{equation*}
\zeta(\bm{u}) := \prod_{\substack{\mathfrak{p}=(e_1\rightharpoonup e_2 \dotsb \rightharpoonup e_n \rightharpoonup e_1)\\ \text{ is a prime cycle}}}
\frac1{\det\left(I_{r_{e_1}}-u_{e_1,e_2}u_{e_2,e_3}\dotsm u_{e_n,e_1}\right)}
\end{equation*}
where $I_r$ is the identity matrix of size $r$.
\end{definition}

If a factor graph includes more than one cycle, the number of prime cycles is infinite.
Hence, it is difficult to evaluate $\zeta(\bm{u})$ from the definition.
The following lemma is generally useful for evaluating $\zeta(\bm{u})$.

\begin{lemma}[Bass's formula]
It holds
$\zeta(\bm{u}) = \det(\mathcal{I}_{|E|}-\mathcal{M}(\bm{u}))^{-1}$
where $\mathcal{I}_{|E|}$ is the identity matrix of size $|E|$ as a block matrix,
and where
$\mathcal{M}(\bm{u})_{e, e'} := u_{e, e'}$ if $e\rightharpoonup e'$
and 
$\mathcal{M}(\bm{u})_{e, e'} := 0$ otherwise.
\end{lemma}
Furthermore, the other expression of $\zeta(\bm{u})$ is known on some condition.
\begin{lemma}[Ihara-Bass formula~\cite{watanabe2010thesis}]
Let $r_i$ be a natural number associated with a variable node $i\in V$.
When $u_{(i\to a), (j\to b)}$ is an $r_i\times r_j$ matrix independent of $b$ and denoted by $u_{i\to j}^a$,
\begin{equation*}
\zeta(\bm{u})^{-1} = \det(\mathcal{I}_{N}-\mathcal{D}+\mathcal{W})
\prod_{a \in F} \det(\mathcal{U}^a)
\end{equation*}
where $\mathcal{D}$ is an $N\times N$ block diagonal matrix defined by
$\mathcal{D}_{i,i} := d_i I_{r_i}$, where $\mathcal{U}^a$ is a $d_a\times d_a$ block matrix defined by
$\mathcal{U}^{a}_{i,i} := I_{r_i}$ and
$\mathcal{U}^{a}_{i,j} := u_{i\to j}^a$ for $i\ne j$, 
and where $\mathcal{W}$ is an $N\times N$ block matrix defined by $\mathcal{W}_{i,j}:= \sum_{a: \{i,j\}\subseteq\partial a} w_{i\to j}^a$.
Here, $w_{i\to j}^a := ({\mathcal{U}^{a}}^{-1})_{i,j}$.
\end{lemma}

\subsection{Determinant of Hessian of the Bethe free energy}
In this subsection, $\{b_i\}$ and $\{b_a\}$ are assumed to be members of an arbitrary fixed parametric family of distributions. 
The alphabet $\mathcal{X}$ is not necessarily finite.
For $i\in V$, $b_i$ has a parameter $\bm{\eta}_i$.
For $a\in F$, $b_a$ has a parameter $\bm{\eta}_a=(\bm{\eta}_{\langle a\rangle}, (\bm{\eta}_i)_{i\in\partial a})$.
The condition~\eqref{eq:marginal} is assumed to be satisfied for any coordinate $(\{\bm{\eta}_i\}, \{\bm{\eta}_{\langle a\rangle}\})$.
In the following, a parameter $\bm{\eta}$ is denoted by the normal font $\eta$ for the simplicity.
For $i\in V$ and $a\in F$, let $H_i$ and $H_a$ be the Shannon entropy of $b_i$ and $b_a$, respectively.
The notation $B\succ 0$ means that a matrix $B$ is positive definite.
\begin{lemma}
For $(\{\eta_i\}, \{\eta_{\langle a\rangle}\})$ satisfying
\begin{align*}
\frac{\partial H_i}{\partial \eta_i \partial \eta_i} &\succ 0,\quad\forall i\in V&
\frac{\partial H_a}{\partial \eta_{\langle a\rangle} \partial \eta_{\langle a\rangle}} &\succ 0,\quad\forall a\in F
\end{align*}
it holds that
\begin{align*}
&\det(\nabla^2 F_\mathrm{Bethe}(\{\eta_i\}, \{\eta_{\langle a\rangle}\})) \\
&=
\prod_{i\in V} \det\Big(\frac{\partial H_i}{\partial \eta_i \partial \eta_i}\Big)
\prod_{a\in F} \det\Big(\frac{\partial H_a}{\partial \eta_{\langle a\rangle} \partial \eta_{\langle a\rangle}}\Big)
\det\big(\mathcal{I}_{N} - \mathcal{D} + \mathcal{G}\big)
\end{align*}
where
\begin{align*}
&\mathcal{G}_{i,j} :=
\bigg(\frac{\partial H_i}{\partial \eta_i\partial \eta_i}\bigg)^{-\frac12}
 \bigg[\sum_{a\in \partial i\cap \partial j}
\bigg(
\frac{\partial H_a}{\partial \eta_i\partial \eta_j}
\\
&-
 \frac{\partial H_a}{\partial \eta_i\partial \eta_{\langle a\rangle}}
\bigg(\frac{\partial H_a}{\partial \eta_{\langle a\rangle}\partial \eta_{\langle a\rangle}}\bigg)^{-1}
 \frac{\partial H_a}{\partial \eta_{\langle a\rangle}\partial \eta_j}
\bigg)\bigg]
\bigg(\frac{\partial H_j}{\partial \eta_j\partial \eta_j}\bigg)^{-\frac12}.
\end{align*}
\end{lemma}
\begin{proof}
It is easy to see that
\begin{align*}
\frac{\partial F_\mathrm{Bethe}}{\partial \eta_i\partial \eta_j} &= 
\sum_{a\in \partial i \cap \partial j} \frac{\partial H_a}{\partial \eta_i\partial \eta_j}
- \delta_{i,j}(d_i-1) \frac{\partial H_i}{\partial \eta_i\partial \eta_i}\\
\frac{\partial F_\mathrm{Bethe}}{\partial \eta_{\langle a\rangle}\partial \eta_{\langle b\rangle}} &= 
\delta_{a,b} \frac{\partial H_a}{\partial \eta_{\langle a\rangle}\partial \eta_{\langle b\rangle}}, \qquad
\frac{\partial F_\mathrm{Bethe}}{\partial \eta_i\partial \eta_{\langle a\rangle}} = 
\frac{\partial H_a}{\partial \eta_i \partial \eta_{\langle a\rangle}}.
\end{align*}
Let $\mathcal{V}$ be a block diagonal matrix defined by
\begin{align*}
\mathcal{V}_{i,i} &:= \frac{\partial H_i}{\partial \eta_i\partial \eta_i},&
\mathcal{V}_{a,a} &:= \frac{\partial H_a}{\partial \eta_{\langle a\rangle}\partial \eta_{\langle a\rangle}}
\end{align*}
and $\mathcal{C}:= \nabla^2 F_\mathrm{Bethe}(\{\eta_i\}, \{\eta_{\langle a\rangle}\}) - \mathcal{V}$.
Then, one obtains
\begin{equation*}
\nabla^2 F_\mathrm{Bethe}(\{\eta_i\}, \{\eta_{\langle a\rangle}\})
= \mathcal{V}^{\frac12}(\mathcal{I}_{N+|F|} + \mathcal{V}^{-\frac12} C\mathcal{V}^{-\frac12})\mathcal{V}^{\frac12}.
\end{equation*}
For $\mathcal{F}:=\mathcal{V}^{-\frac12} C\mathcal{V}^{-\frac12}$, it holds that
\begin{align*}
\mathcal{F}_{i,j} &= \mathcal{V}_{i,i}^{-\frac12}\sum_{a\in\partial i \cap \partial j} 
\frac{\partial H_a}{\partial \eta_i\partial \eta_j}\mathcal{V}_{j,j}^{-\frac12} - \delta_{i,j} d_i I_{r_i},\hspace{3em}
\mathcal{F}_{a,b} = 0\\
\mathcal{F}_{i,a} &= 
\mathcal{V}_{i,i}^{-\frac12} \frac{\partial H_a}{\partial \eta_i\partial \eta_{\langle a\rangle}}\mathcal{V}_{a,a}^{-\frac12},
\hspace{2.5em}
\mathcal{F}_{a,i} = 
\mathcal{V}_{a,a}^{-\frac12} \frac{\partial H_a}{\partial \eta_{\langle a\rangle}\partial \eta_i}\mathcal{V}_{i,i}^{-\frac12}.
\end{align*}
From
$\det(\nabla^2 F_\mathrm{Bethe}(\{\eta_i\}, \{\eta_{\langle a\rangle}\})) = \det(\mathcal{V}) \det(\mathcal{I}_{N+|F|} + \mathcal{F})$
and
\begin{equation*}
\det(\mathcal{V})
=\prod_{i\in V} \det\Big(\frac{\partial H_i}{\partial \eta_i \partial \eta_i}\Big)
\prod_{a\in F} \det\Big(\frac{\partial H_a}{\partial \eta_{\langle a\rangle} \partial \eta_{\langle a\rangle}}\Big)
\end{equation*}
we only have to prove
$\det(\mathcal{I}_{N+|F|} + \mathcal{F})
= \det(\mathcal{I}_{N}-\mathcal{D}+\mathcal{G})$.
For $u\times u$, $u\times v$ and $v\times u$ matrices $A$, $B$ and $C$, respectively,
it holds
$
\begin{bmatrix}
A& B\\
C& I_v
\end{bmatrix}
\begin{bmatrix}
I_u& 0\\
-C& I_v
\end{bmatrix}
=
\begin{bmatrix}
A-BC& B\\
0& I_v
\end{bmatrix}
$
and hence
$
\det\bigg(
\begin{bmatrix}
A& B\\
C& I_v
\end{bmatrix}
\bigg)
=\det(A-BC)$.
Therefore,
\begin{align*}
\det(\mathcal{I}_{N+|F|}+\mathcal{F}) &= \det(\mathcal{I}_{N}+\mathcal{F}_\mathrm{VV}-\mathcal{F}_\mathrm{VF}\mathcal{F}_\mathrm{VF}^t)\\
&= \det(\mathcal{I}_{N}-\mathcal{D}+\mathcal{G}).\qedhere
\end{align*}
\end{proof}

For an exponential family, the determinant of Hessian of the Bethe free energy is connected to
the edge zeta function.
\begin{lemma}[Watanabe-Fukumizu formula~\cite{watanabe2010thesis}]
Let $\{\eta_i\}$ and $\{\eta_a\}$ be the expectation parameters for $\{b_i\}$ and $\{b_a\}$, respectively.
Let $\{t_i\}$ and $\{t_a\}$ be the sufficient statistics for $\{b_i\}$ and $\{b_a\}$, respectively.
Then, it holds
\begin{align*}
\zeta(\bm{u})^{-1}&=
\det(\nabla^2 F_\mathrm{Bethe}(\{\eta_i\}, \{\eta_{\langle a\rangle}\}))\\
&\quad\cdot \prod_{i\in V} \det(\Var_{b_i}[t_i(X_i)])^{1-d_i}
\prod_{a\in F} \det(\Var_{b_a}[t_a(X_{\partial a})])^{-1}
\end{align*}
where 
$r_i$ is the number of parameters of $b_i$ for $i\in V$ and where
\begin{align*}
u_{i\to j}^a &= 
\Cor_{b_a}[t_i(X_i), t_j(X_j)]\\
&:=
\Var_{b_i}[t_i(X_i)]^{-\frac12}
\Cov_{b_a}[t_i(X_i), t_j(X_j)]
\Var_{b_j}[t_j(X_j)]^{-\frac12}.
\end{align*}
Here, $\Cov_{b_a}[t_i(X_i), t_j(X_j)]$ is a matrix whose $(k, l)$-element is
\begin{equation*}
\mathbb{E}[(t_{i,k}(X_i) - \mathbb{E}[t_{i,k}(X_i)])(t_{j,l}(X_j)-\mathbb{E}[t_{j.l}(X_j))].
\end{equation*}
\end{lemma}

\section{The new approximations via asymptotic expansion}\label{sec:new}
\subsection{The series of approximations}
Recently, a new characterization of the Bethe free energy is shown by Vontobel~\cite{vontobel2010counting}.
Let $Z(M)$ be the random variable
corresponding to the partition function of a uniformly chosen random graph cover of the original factor graph
where $M$ is the number of copies of the original graph in the graph cover.
See~\cite{vontobel2010counting} for a detailed definition of a graph cover.
Then, the Bethe free energy appears naturally as follows.
\begin{lemma}[\cite{vontobel2010counting}]
It holds
$\mathbb{E}[Z(M)] = Z_\mathrm{Bethe}^{M + o(M)}$ as $M\to \infty$.
\end{lemma}
This result gives a new characterization of the Bethe free energy.
More detailed analysis is obtained as follows.
\begin{lemma}[\cite{mori2012central}]
Let $\mathcal{S}(f)$ be the support of a function $f$.
Let $\mathcal{B}$ be the set of minima of $F_\mathrm{Bethe}(\{b_i\}, \{b_a\})$.
Assume $|\mathcal{B}|<\infty$ and
\begin{align*}
&b^*_i(x)>0,\; \forall x\in\mathcal{X},\; \forall i\in V,\quad
b^*_a(\bm{x})>0,\; \forall\bm{x}\in\mathcal{S}(f_a),\; \forall a\in F\\
&\det\left(\nabla^2F_\mathrm{Bethe}(\{\eta^*_i\}, \{\eta^*_{\langle a\rangle}\})\right) >0
\end{align*}
for all $(\{b^*_i\}, \{b^*_a\})\in\mathcal{B}$.
Then, it holds that as $M\to\infty$
\begin{align*}
&\mathbb{E}[Z(M)] = Z_\mathrm{Bethe}^M\sum_{(\{b^*_i\}, \{b^*_a\})\in\mathcal{B}}
\Bigg(\det\left(\nabla^2F_\mathrm{Bethe}(\{\eta^*_i\}, \{\eta^*_{\langle a\rangle}\})\right)\\
&\quad\cdot \prod_{i\in V}\prod_{x\in\mathcal{X}}b_i^*(x)^{1-d_i}\prod_{a\in F}\prod_{\bm{x}\in\mathcal{S}(f_a)}b_a^*(\bm{x})\Bigg)^{-\frac12}
(1+o(1))\\
&=
Z_\mathrm{Bethe}^M\sum_{(\{b^*_i\}, \{b^*_a\})\in\mathcal{B}} \sqrt{\zeta(\bm{u})}(1+o(1))
\end{align*}
where $u_{i\to j}^a = \Cor_{b^*_a}[t_i(X_i), t_j(X_j)]$
and where 
$\eta_i=\{b_i(x)\}_{x\in\mathcal{X}\setminus 0}$
and
$\eta_{\langle a\rangle}=T_a(\{b_a(\bm{x})\}_{\bm{x}\in\mathcal{S}(f_a)})$.
Here, for $a\in F$, $T_a(\eta)$ denotes an any subset of $\eta$ with which $\{\eta_i\}_{i\in\partial a}$
can be regarded as an expectation parameter for $b_a$.
\end{lemma}

Furthermore, except for cases including singularity, the expectation of a partition function of a graph cover has the asymptotic expansion
\begin{align*}
\log \mathbb{E}[Z(M)] \sim M \log Z_\mathrm{Bethe} + \sum_{k=0}^\infty \frac{g_k}{\sqrt{M^k}}
\end{align*}
where
$g_0 := \log\sum_{(\{b^*_i\}, \{b^*_a\})\in\mathcal{B}} \sqrt{\zeta(\bm{u})}$
and where $\{g_k\}_{k=1,2,\dotsc,}$ are some constants similarly to the Edgeworth expansion~\cite{chambers1967methods}.
Based on the above divergent series, we demonstrate the following series of approximations by letting $M=1$, which may seem shameful
\begin{align*}
\log Z &\approx \log Z_\mathrm{Bethe}\\
\log Z &\approx \log Z_\mathrm{Bethe} + g_0\\
\log Z &\approx \log Z_\mathrm{Bethe} + g_0 + g_1\\
&\vdots
\end{align*}

\begin{definition}[Asymptotic Bethe approximation]
For $m=0,1,\dotsc$, the asymptotic Bethe approximation of order $m$ is defined as
\begin{equation*}
Z_{\mathrm{AB}}^{(m)} := Z_\mathrm{Bethe}\exp\bigg\{\sum_{k=0}^{m-1} g_k\bigg\}.
\end{equation*}
\end{definition}

\subsection{Asymptotic exactness of the asymptotic Bethe approximation of order 1}
In this subsection, we show cases in which the asymptotic Bethe approximation $Z_{\mathrm{AB}}^{(1)}$ is asymptotically better in some limit than 
the Bethe approximation $Z_\mathrm{Bethe}$.
Let the set $L_2$ of loops be
\begin{align*}
L_2 := \big\{E'\subseteq E\mid \text{$E'$ is connected }, d_o(E') = 0 \text{ or }2,\\
 \forall o\in V\cup F\big\}.
\end{align*}
The set $L_2$ of loops is a subset of the set $G$ of generalized loops.
For the binary case, a contribution $\mathcal{Z}(E')$ in Lemma~\ref{lem:loop} of a loop $E'\in L_2$ is
\begin{align*}
\mathcal{Z}(E')=
\Cor_{b_{a_\ell}}[X_{i_\ell}, X_{i_{1}}]
\prod_{k=1}^{\ell-1}
\Cor_{b_{a_k}}[X_{i_k}, X_{i_{k+1}}]
\end{align*}
where $(i_\ell\to a_\ell)\rightharpoonup(i_1\to a_1)\rightharpoonup(i_2\to a_2)\rightharpoonup \dotsb \rightharpoonup(i_\ell\to a_\ell)$
forms the loop $E'$ and where $\ell:= |E'|$.
In the above, weight $\Cor_{b_a}[X_i, X_j]$ of the connection of edges $(i_k\to a_k)\rightharpoonup (i_{k+1}\to a_{k+1})$
is the same as $u_{i_k\to j_{k+1}}^{a_k}$ in the edge zeta function used for $Z_{\mathrm{AB}}^{(1)}$ for $k=1,\dotsc,\ell$.

For non-binary cases, the partition function of a single-cycle graph can also be calculated easily
by using the well-known method of the transfer matrix.
\begin{lemma}[Transfer matrices]
\begin{equation*}
\sum_{\bm{x}\in\mathcal{X}^N}
f_N(x_N, x_1) \prod_{i=1}^{N-1} f_i(x_i, x_{i+1})
= \mathrm{tr}(F^{(1)} F^{(2)} \dotsm F^{(N)})
\end{equation*}
where $F^{(i)}$ is a $|\mathcal{X}|\times |\mathcal{X}|$ matrix with $F^{(i)}_{x, x'} = f_i(x, x')$.
\end{lemma}
From this lemma and Lemma~\ref{lem:fisher},
one obtains $\mathcal{Z}(E')$
in Theorem~\ref{thm:non-binary}
for $E'\in L_2$ as
\begin{align}
\mathcal{Z}(E')&=
\mathrm{tr}
\big(
\Cor_{b_{a_1}}[t_{i_1}(X_{i_1}), t_{i_2}(X_{i_2})]
\Cor_{b_{a_2}}[t_{i_2}(X_{i_2}), t_{i_3}(X_{i_3})]\nonumber\\
&\qquad \dotsm
\Cor_{b_{a_{\ell}}}[t_{i_\ell}(X_{i_\ell}), t_{i_1}(X_{i_1})]
\big).
\label{eq:mloop}
\end{align}

First, we consider the simplest non-trivial example, namely, a single-cycle factor graph.
\begin{example}[Single-cycle factor graph]
For a single-cycle graph, the Bethe free energy is convex with respect to the expectation parameters
and hence the stationary point is unique~\cite{watanabe2010thesis}.
For the unique solution $\{b^*_i\}, \{b^*_a\}$, one obtains from Theorem~\ref{thm:non-binary} and~\eqref{eq:mloop} that
\begin{equation*}
\frac{Z}{Z_\mathrm{Bethe}(\{b^*_i\}, \{b^*_a\})} = 1+
\mathrm{tr}(A)
\end{equation*}
where
\begin{align*}
A:=
\Cor_{b^*_{a_1}}[t_1(X_1), t_2(X_2)]
\Cor_{b^*_{a_2}}[t_2(X_2), t_3(X_3)]\\
\dotsm
\Cor_{b^*_{a_N}}[t_N(X_N), t_1(X_1)].
\end{align*}
On the other hand, the square root of the edge zeta function is
\begin{align*}
\sqrt{\zeta(\bm{u})}
&= 
\frac1{\det\big(I_{|\mathcal{X}|-1}- A\big)}.
\end{align*}
From $\det(I_{|\mathcal{X}|-1}-A)=1-\mathrm{tr}(A) + O(\rho(A)^2)$ as $A\to 0$,
where $\rho(A)$ denotes the spectral radius of $A$, 
one obtains the following asymptotic equality
\begin{align*}
\sqrt{\zeta(\bm{u})}
&=
\frac1{1-\mathrm{tr}(A) + O(\rho(A)^2)}\\
&=
1+\mathrm{tr}(A)+ O(\rho(A)^2)
=\frac{Z}{Z_\mathrm{Bethe}}+ O(\rho(A)^2).
\end{align*}
Hence, $\sqrt{\zeta(\bm{u})}$ is an accurate approximation for $Z/Z_\mathrm{Bethe}$
whenever the matrix $A$ is close to zero. 
\hfill\qed
\end{example}
For a factor graph including multiple cycles,
we consider the case that the functions $\{f_a\}$ have a parameter $\beta>0$.
We assume that there is a unique minimum of the Bethe free energy.  
Furthermore, we assume that $Z/Z_\mathrm{Bethe}\to 1$ as $\beta\to 0$, and that
the contribution of loops dominates $Z/Z_\mathrm{Bethe}$, i.e.,
\begin{align}
\frac{Z}{Z_\mathrm{Bethe}}
&= 1 + \sum_{E'\in L_2\setminus \{\varnothing\}} \mathcal{Z}(E')
+ o\bigg(\sum_{E'\in L_2\setminus \{\varnothing\}} \mathcal{Z}(E')\bigg)
\label{eq:asm}
\end{align}
as $\beta\to 0$.
On the other hand, one obtains
\begin{align*}
\sqrt{\zeta(\bm{u})}
&= 1 + \sum_{E'\in L_2\setminus \{\varnothing\}} \mathcal{Z}(E')
+ o\bigg(\sum_{E'\in L_2\setminus \{\varnothing\}} \mathcal{Z}(E')\bigg).
\end{align*}
Hence, the square root of the edge zeta function $\sqrt{\zeta(\bm{u})}$ gives the dominant terms in $Z/Z_\mathrm{Bethe}$
in the limit $\beta\to 0$ for an arbitrary topological graph.
The derivation of the above equation is omitted due to the lack of the space.

\begin{example}[High-temperature expansion for the Ising model]
For $J_{i,j}\in\mathbb{R}$ and $\beta>0$, the partition function of the Ising model is
\begin{equation}\label{eq:ising}
Z = \sum_{\bm{x}\in\{+1,-1\}^N}\exp\bigg\{\beta\bigg(\sum_{(i,j)\in F}J_{i,j}x_ix_j + \sum_{i\in V} h_i x_i\bigg) \bigg\}.
\end{equation}
When $h_i = 0$ for all $i\in V$, 
the set of uniform messages is the trivial solution of the BP equations~\eqref{eq:bp}.
Although it is not the unique solution of~\eqref{eq:bp},
if $\beta$ is sufficiently small, the set of uniform messages is the unique minimum of the Bethe free energy.
The Fourier transform is equivalent to the Bethe transform for the uniform messages.
From Lemma~\ref{lem:loop}, or equivalently from the MacWilliams identity,
one obtains
\begin{align*}
Z
&= 2^N\prod_{(i,j)\in F}\cosh(\beta J_{i,j})
\sum_{\substack{E'\subseteq E,\\ d_i(E') \text{ is even}}}\prod_{(i,j)\in E'}\tanh(\beta J_{i,j}).
\end{align*}
In this case, the Bethe free energy evaluated with the uniform messages is $2^N\prod_{(i,j)\in F} \cosh(\beta J_{i,j})$.
Since~\eqref{eq:asm} is satisfied in this case, $Z_{\mathrm{AB}}^{(1)}$ is a more accurate approximation than $Z_\mathrm{Bethe}$
in the limit $\beta\to 0$.
When $h_i\ne 0$, the uniform messages are not a solution of the BP equations~\eqref{eq:bp}.
In that case, from~\eqref{eq:biba},
$\Cor_{b_a}[X_i, X_j]$ is equal to
\begin{align*}
&\frac{\sinh(2\beta J_{i,j})}{\sqrt{\cosh(2l_i) + \cosh(2\beta J_{i,j})}\sqrt{\cosh(2l_j) + \cosh(2\beta J_{i,j})}}
\end{align*}
where $m_{i\to a}\propto \exp\{l_i\}$ and $m_{j\to a}\propto \exp\{l_j\}$~\cite{watanabe2010thesis}.
Since $|\Cor_{b_a}[X_i, X_j]|$ takes the maximum $|\tanh(\beta J_{i,j})|$ at $l_i=l_j=0$,
as $\beta\to 0$, $\Cor_{b_a}[X_i, X_j]\to 0$ and hence,~\eqref{eq:asm} is satisfied.
On the other hand, as $h_i\to\infty$, it also holds $\Cor_{b_a}[X_i, X_j]\to 0$ since $l_i\to\pm\infty$ as $h_i\to\pm\infty$.
Hence, the condition~\eqref{eq:asm} is also satisfied in the limit $h_i\to\pm\infty$ for all $i\in V$.
\hfill\qed
\end{example}
\begin{remark}
Let us consider the (not bipartite) graph in which a degree-two factor node in the Ising model~\eqref{eq:ising} is replaced by an edge.
Due to Kac and Ward~\cite{kac1952combinatorial},
if the graph is planar and if $h_i=0$ for all $i\in V$, the partition function $Z$ can be calculated in polynomial time from the beautiful equation
$Z=Z_{\mathrm{Bethe}}/\sqrt{\zeta(\bm{u})}$
where $u_{(i\to a), (j\to b)}=\tanh(\beta J_{i,j})\exp\{\sqrt{-1}\gamma_{(i\to a), (j\to b)}/2\}$.
Here, $\gamma_{(i\to a), (j\to b)}$ denotes the angle of the edge connection $(i\to a)\rightharpoonup (j\to b)$.
This equation exists only when the graph is a certain type of planar graph and the magnetic field $h_i$ is zero for all $i\in V$.
\hfill \qed
\end{remark}
Even if each $\Cor_{b_a}[X_i, X_j]$ does not go to 0, if the product of them along a loop goes to 0,
the condition~\eqref{eq:asm} can be satisfied.
Some sparse factor graphs can satisfy this condition in the large-size limit $N\to\infty$.
The justification of the approximation $Z_{\mathrm{AB}}^{(1)}$ for some sparse factor graphs is an open problem.

When the Hessian of the Bethe free energy is not positive definite at some critical temperature $\beta_\mathrm{c}$,
the edge zeta function diverges.
This situation is considered as a finite-size analogue of the second-order phase transition.
Similarly, if the minima of the Bethe free energy discontinuously jumps at $\beta_\mathrm{c}$,
$Z_\mathrm{AB}^{(m)}$ is discontinuous at $\beta_\mathrm{c}$ for $m\ge 1$.
This situation is considered as a finite-size analogue of the first-order phase transition.
In these cases, it is better to consider another limit for $\beta$ around $\beta_\mathrm{c}$,
e.g., $\delta=(\beta-\beta_\mathrm{c})/N^c$ is fixed for some $c>0$~\cite{parisi1993critical}.

For the three-body Ising model without magnetic field, $\sqrt{\zeta(\bm{u})}=1$ for the uniform messages
and hence $Z_\mathrm{AB}^{(1)}=Z_\mathrm{Bethe}$.
We can intuitively guess that $Z_\mathrm{AB}^{(2)}$ takes account of contributions of
the set of connected generalized loops in which all degrees must be at most three
since the derivation of $g_1$ requires derivation of third-order derivatives of the Bethe free energy, which includes
the third-order statistics in Theorem~\ref{thm:non-binary}.
This idea can be generalized to $Z_{\mathrm{AB}}^{(m)}$ for any $m\ge 1$.

\bibliographystyle{IEEEtran}
{\footnotesize
\bibliography{IEEEabrv,ldpc}
}
\thispagestyle{mysimple}

\end{document}